\documentclass[11pt]{article}
\usepackage[T1]{fontenc}
\usepackage{amsfonts}
\usepackage{amsmath}
\usepackage{amssymb}
\usepackage{amsthm}
\usepackage{bbm}
\usepackage{bm}
\usepackage{mathrsfs}
\usepackage{verbatim}
\usepackage{setspace}
\usepackage{color}
\usepackage{pdfsync}
\usepackage{enumitem}

\theoremstyle{plain}
\newtheorem{theorem}{Theorem}[section]
\newtheorem{proposition}[theorem]{Proposition}
\newtheorem{lemma}[theorem]{Lemma}

\theoremstyle{definition}
\newtheorem{definition}[theorem]{Definition}
\newtheorem{remark}[theorem]{Remark}

\newtheorem{example}[theorem]{Example}

\newtheorem*{conditionA}{Condition~(A)}
\theoremstyle{remark}

\renewenvironment{thebibliography}[1]{%
\begin{oldthebibliography}{#1}%
\setlength{\baselineskip}{.9em}
\linespread{1}
\small
\setlength{\parskip}{0ex}%
\setlength{\itemsep}{.29em}
}%
{%
\end{oldthebibliography}%
}

\newcommand{\F}{\mathbb{F}}
\newcommand{\G}{\mathbb{G}}

\newcommand{\Q}{\mathbb{Q}}
\newcommand{\R}{\mathbb{R}}
\renewcommand{\S}{\mathbb{S}}

\newcommand{\cE}{\mathcal{E}}
\newcommand{\cF}{\mathcal{F}}
\newcommand{\cG}{\mathcal{G}}
\newcommand{\cH}{\mathcal{H}}
\newcommand{\cL}{\mathcal{L}}

\newcommand{\fP}{\mathfrak{P}}

\DeclareMathOperator{\diag}{diag}

\DeclareMathOperator{\tr}{tr}

\DeclareMathOperator{\esssup}{ess\, sup}

\newcommand{\as}{\mbox{-a.s.}}

\newcommand{\1}{\mathbf{1}}
\newcommand{\sint}{\stackrel{\mbox{\tiny$\bullet$}}{}}
\newcommand{\br}[1]{\langle #1 \rangle}

\newcommand{\bomega}{\bar{\omega}}

\numberwithin{equation}{section}

\usepackage[pdfborder={0 0 0}]{hyperref}
\hypersetup{
  urlcolor = black,
  pdfauthor = {Marcel Nutz},
  pdfkeywords = {Optional decomposition; Model uncertainty; Superreplication; Nondominated model},
  pdftitle = {Robust Superhedging with Jumps and Diffusion},
  pdfsubject = {Robust Superhedging with Jumps and Diffusion},
  pdfpagemode = UseNone
}

\begin{document}

\title{\vspace{-0em}Robust Superhedging with Jumps and Diffusion}
\date{\today}
\author{
  Marcel Nutz%
  \thanks{
  Departments of Statistics and Mathematics, Columbia University, New York, \texttt{mnutz@columbia.edu}. Financial support by NSF Grants DMS-1208985 and DMS-1512900 is gratefully acknowledged. The author thanks Kostas Kardaras, Ariel Neufeld, Nizar Touzi and Jianfeng Zhang for fruitful discussions, and he is grateful to the Associate Editor and two anonymous referees for their constructive comments.
  }
 }

\maketitle \vspace{-1em}

\begin{abstract}
We establish a nondominated version of the optional decomposition theorem in a setting that includes jump processes with nonvanishing diffusion as well as general continuous processes. This result is used to derive a robust superhedging duality and the existence of an optimal superhedging strategy for general contingent claims. We illustrate the main results in the framework of nonlinear L\'evy processes.
\end{abstract}

\vspace{.9em}

{\small

}

\section{Introduction}\label{se:intro}

The classical optional decomposition theorem states that given a process~$Y$ which is a supermartingale under all equivalent martingale measures of some reference process $S$, there exists an integrand $H$ such that $Y-H\sint S$ is nonincreasing, where $H\sint S$ denotes the stochastic integral. Stated differently, $Y$ admits the decomposition $Y=Y_{0}+H\sint S -K$ for some nondecreasing process~$K$. This result is stated on a given probability space $(\Omega,\cF,P_{*})$; without loss of generality, one can assume that $S$ is itself a $P_{*}$ martingale.
The optional decomposition theorem is due to~\cite{ElKarouiQuenez.95} in the case of a continuous process $S$, while the case with jumps is due to~\cite{Kramkov.96} under a boundedness assumption and~\cite{FollmerKabanov.98} in the general case. An alternative proof was presented in~\cite{DelbaenSchachermayer.99}, and~\cite{FollmerKramkov.97} extended the result to include portfolio constraints. Optional decomposition theorems are important for applications in mathematical finance, in particular for superreplication pricing and portfolio optimization.

The first result of this paper (Theorem~\ref{th:optDecomp}) is a version of the optional decomposition theorem which is suitable in the context of model uncertainty: it does not require a reference measure. More precisely, we consider a set $\fP$ of probabilities, possibly nondominated in the sense that its elements are not dominated by a single reference probability $P_{*}$. Suppose that $S$ is a c\`adl\`ag local martingale under all elements of $\fP$, and that $\fP$ contains all equivalent local martingale measures of its elements. If $Y$ is a c\`adl\`ag supermartingale under all $P\in\fP$, we show that there exists an integrand $H$ such that $Y-H\sint S$ is nonincreasing $P$-a.s.\ for all $P\in\fP$. This result is obtained under a technical condition that we call dominating diffusion property (Definition~\ref{de:dominatingDiffusion}): for all $P\in\fP$, the jump characteristic $\nu^{P}$ of $S$ is dominated by the diffusion characteristic $C^{P}$. This includes the case where $S$ is an It\^{o} semimartingale with jumps and nonvanishing diffusion, as well as the case of a general continuous process. The dominating diffusion property allows us to define $H$ in terms of the joint diffusion characteristic of $(S,Y)$ and thus to take advantage of the fact that the latter can be constructed in an aggregated way (i.e., simultaneously for all $P$). The proof that this strategy is superreplicating capitalizes on the classical optional decomposition theorem under each $P\in\fP$. Thus, the argument is quite simple, with the advantage that the result is general and the proof is versatile (see \cite{BiaginiBouchardKardarasNutz.14} for an adaptation to a different setting). In particular, we do not require delicate separability conditions on the filtration or compactness assumptions on $\fP$.

The second result (Theorem~\ref{th:duality}) is a superhedging duality in the setting of model uncertainty. In a setup where Skorohod space is used as the underlying measurable space and the set $\fP$ satisfies certain dynamic programming conditions, it is shown that given a measurable function $f$ at the time horizon $T$,  the robust superhedging price 
\[
  \pi(f):=\inf\big\{x\in\R:\,\exists\, H\mbox{ with } x+ H\sint S_{T} \geq f\;P\as\mbox{ for all }P\in\fP\big\}
\]
satisfies the duality relation 
\[
 \pi(f)=\sup_{P\in\fP} E^P[f].
\]
Moreover, the infimum is attained; i.e., an optimal superhedging strategy exists. We construct this strategy through the above optional decomposition result.

Finally, we illustrate our main results in the setting of nonlinear L\'evy processes; this is a natural example where the set $\fP$ is defined in terms of the characteristics of $S$. We characterize the conditions of our main results in terms of the model primitives and discuss further aspects of our problem formulation.

The main novelty in our results is the applicability to nondominated continuous-time models with jumps. To the best of our knowledge, there is no extant result providing the existence of an optimal strategy or an optional decomposition theorem in this framework. The only previous result is the duality statement of \cite{DolinskySoner.14} in the context of optimal transport; there, the absence of a duality gap is proved by a weak approximation with discrete models and the superreplication is formulated in a pathwise fashion. In ongoing independent work~\cite{CheriditoKupperTangpi.14}, absence of a duality gap will be established by functional analytic methods, though under a compactness condition in Skorohod space which is generally not satisfied in our setting.

The case of continuous processes (i.e., volatility uncertainty) is better studied. The duality formula in this context is investigated by \cite{DenisMartini.06} from a capacity-theoretic point of view, \cite{Peng.10, SonerTouziZhang.2010rep, Song.10} use an approximation by Markovian control problems, and \cite{DolinskySoner.12} uses a weak approximation based on dominated models. On the other hand, \cite{NeufeldNutz.12, NutzSoner.10, NutzZhang.12, PossamaiRoyerTouzi.13, SonerTouziZhang.2010dual} use an aggregation argument which can be seen as a predecessor of our proof; however, they rely on the Doob--Meyer decomposition theorem (under each $P\in\fP$) and this forces them to assume that each $P\in\fP$ corresponds to a complete market. 
Thus, the present results are an important improvement even in the continuous case, as they apply also with incomplete markets.

In the discrete-time case, \cite{BouchardNutz.13} provides a general version of the optional decomposition theorem.
That result is not recovered by the present paper because the dominating diffusion condition is not satisfied. An earlier duality result, without existence of an optimal strategy, is presented by \cite{AcciaioBeiglbockPenknerSchachermayer.12} in a different discrete-time setup; see also \cite{DolinskySoner.13} for related results in the context of transaction costs, \cite{BayraktarZhou.14} for portfolio constraints, \cite{Dolinsky.13} for game options, and \cite{BayraktarHuangZhou.13} for American options.

The remainder of this paper is organized as follows. Section~\ref{se:optionalDecomp} deals with the optional decomposition in a general setting, Section~\ref{se:duality} establishes the duality result on Skorohod space, and Section~\ref{se:Levy} concludes with the example of nonlinear L\'evy processes. 

\section{Optional Decomposition}\label{se:optionalDecomp}

Let $T>0$ and let $(\Omega,\cF)$ be a measurable space equipped with an arbitrary filtration $\F=(\cF_{t})_{t\in[0,T]}$. Let $S=(S_{t})$ be an $\R^{d}$-valued, $\F$-adapted process with c\`adl\`ag paths, for some positive integer $d$. We denote by $\fP(\Omega)$ the set of all probability measures on $(\Omega,\cF)$. Given $P\in\fP(\Omega)$, we write $\F^{P}_{+}$ for the usual $P$-augmentation of $\F$. Moreover, if $S$ is a semimartingale under $P\in\fP(\Omega)$ and $H$ is an $\F^{P}_{+}$-predictable, $d$-dimensional integrand, we denote by $H\sint S_{t}=\int_{0}^{t} H_{u}\, dS_{u}$ the usual stochastic (It\^o) integral under $P$.

A probability $P\in\fP(\Omega)$ is called a \emph{sigma martingale measure for $S$} if $S$ is a sigma martingale with respect to $(P,\F^{P}_{+})$. We recall the definition of a sigma martingale: there exist $\F^{P}_{+}$-predictable sets $(\Sigma_{n})_{n\geq1}$ increasing to $\Omega\times [0,T]$ such that $\1_{\Sigma_{n}}\sint S^{i}$ is a martingale for each $n$ and each component $S^{i}$ of $S$.
(While everything in this section is true also for local martingales, this generalization of the local martingale property will be convenient later on.)
We shall say that $\fP$ is \emph{saturated} if it contains all equivalent sigma martingale measures of its elements. That is, if $P'\in\fP(\Omega)$ is a sigma martingale measure for $S$ and $P'\sim P$ for some $P\in\fP$, then $P'\in\fP$.
Finally, a real-valued, $\F$-adapted process with c\`adl\`ag paths is called a \emph{$\fP$ local supermartingale} if it is a local supermartingale with respect to $(P,\F^{P}_{+})$ for all $P\in\fP$. We refer to \cite{JacodShiryaev.03} for background on stochastic calculus and unexplained notation.

\begin{remark}\label{rk:filtrationDoesntMatter}
  The choice of the filtration $\F^{P}_{+}$ in the above definitions is the most general one: if $X$ is a right-continuous $\F$-adapted process which is a $P$ local supermartingale (or sigma martingale) with respect to $(P,\tilde\F)$ for some filtration $\F\subseteq \tilde\F \subseteq \F^{P}_{+}$, then it has the same property with respect to $(P,\F^{P}_{+})$. This follows from the backward martingale convergence theorem.
\end{remark}

Fix a truncation function $h: \R^d\to \R^d$; that is, a bounded measurable function such that $h(x)=x$ in a neighborhood of the origin. Given $P\in\fP(\Omega)$ under which $S$ is a semimartingale, we denote by $(B^P,C^P,\nu^P)$ the semimartingale characteristics of~$S$ under $P$, relative to $h$. That is, $(B^P,C^P,\nu^P)$ is a triplet of processes such that $P$-a.s., $B^P$ is the finite variation part in the canonical decomposition of $S -\sum_{0 \leq s \leq \cdot} (\Delta S_s- h(\Delta S_s))$ under~$P$, $C^P$ is the quadratic covariation of the continuous local martingale part of~$S$ under~$P$, and $\nu^P$ is the $P$-compensator of $\mu^S$, the integer-valued random measure associated with the jumps of $S$. (Again, the precise choice of the filtration does not matter for the present purposes; cf.\ \cite[Proposition~2.2]{NeufeldNutz.13a}.)

Next, we introduce a notion that will play a key role in our proofs.

\begin{definition}\label{de:dominatingDiffusion}
	Let $P\in\fP(\Omega)$ be a sigma martingale measure for $S$ and let $(B^{P},C^{P},\nu^{P})$ be semimartingale characteristics of $S$ under $P$. We say that $S$ has \emph{dominating diffusion} under $P$ if
  \[
    \nu^{P} \ll (C^{P})^{ii} \quad  P\as,\quad i=1,\dots,d.
  \]
\end{definition}

The notation $\nu^{P} \ll (C^{P})^{ii}$ means that the process $(|x|^{2}\wedge 1) \ast \nu^{P}_{t}:= \int_{0}^{t}\int_{\R^{d}} (|x|^{2}\wedge 1) \,\nu^{P}_{s}(dx,ds)$ is absolutely continuous with respect to the $i$th component on the diagonal of the matrix $C^{P}$. We remark that the first characteristic then necessarily satisfies $B^P\ll (C^{P})^{ii}$ $P$-a.s., because the sigma martingale property implies that $B^{P} = \int (x-h(x))\ast \nu^{P}$; cf.\ \cite[Proposition~III.6.35, p.\,215]{JacodShiryaev.03}. To illustrate the significance of the definition, consider the one-dimensional case $d=1$ for simplicity. Then, it follows that the measure $P\times dC^{P} $ is rich enough to express the properties of $S$ that will be relevant in what follows; in particular, if $H$ and $H'$ are predictable integrands such that $H=H'$ $P\times dC^{P}$-a.e., then $H\sint S=H'\sint S$ $P$-a.s.

The preceding definition is satisfied in the following important cases; see also Lemma~\ref{le:recipe}.

\begin{example}\label{ex:assumptSatisfied}
  (i)  Let $S$ be a sigma martingale with absolutely continuous characteristics under $P$ (with respect to the Lebesgue measure $dt$); i.e., the characteristics are of the form $(dB,dC,d\nu)=(b\,dt,c\,dt,F\,dt)$. If $c$ is a strictly positive matrix $P\times dt$-a.e., then $S$ has dominating diffusion.  

  (ii) Let $S$ be a \emph{continuous} sigma martingale under $P$ (hence a local martingale). Then its characteristics are of the form $(0,C,0)$ and thus $S$ always has dominating diffusion.
\end{example}

The dominating diffusion property does not hold for pure-jump models and in particular for the discrete-time case. However, as mentioned in the Introduction, the latter is covered by the results of~\cite{BouchardNutz.13}. In view of those results, we see the dominating diffusion property as a technical assumption that enables us to argue as in the proof stated below, but we do not expect generic counterexamples to the subsequent theorem even if the property is violated. Nevertheless, \cite{BouchardNutz.13} suggests that other technical assumptions may be necessary and that a proof without dominating diffusion may be substantially less transparent. 

We can now state our first result, a nondominated version of the optional decomposition theorem. We denote by $L(S,\fP)$ the set of all $\R^{d}$-valued, $\F$-predictable processes which are $S$-integrable for all $P\in\fP$. 

\begin{theorem}\label{th:optDecomp}
  Let $\fP$ be a nonempty, saturated set of sigma martingale measures for $S$ such that $S$ has dominating diffusion under all $P\in\fP$. If $Y$ is a $\fP$ local supermartingale, then there exists an $\F$-predictable process $H\in L(S,\fP)$ such that
  \[
    Y-H\sint S \quad \mbox{is nonincreasing $P$-a.s.}
    \footnote{Here $H\sint S$ is the usual It\^o integral under the fixed measure $P$, and this is sufficient for our purposes. If desired, then under the conditions of the theorem, it is also possible to define the integral in a pathwise manner and in particular simultaneously under all $P\in\fP$, by the construction of~\cite{Nutz.11int}.
}   
  \quad \mbox{for all} \quad P\in\fP.
  \]
\end{theorem}

\begin{proof}
	We say that a property holds $\fP$-q.s.\ if it holds $P$-a.s.\ for all $P\in\fP$. The $(d+1)$-dimensional process $(S,Y)$ is a semimartingale under all $P\in\fP$ (in the filtration $\F^{P}_{+}$, but then also in $\F$; cf.\ \cite[Proposition~2.2]{NeufeldNutz.13a}). According to \cite[Proposition~6.6]{NeufeldNutz.13a}\footnote{That proposition does not use the separability assumptions on the filtration that are imposed for the main results of \cite{NeufeldNutz.13a}.}, there exists an $\F$-predictable process $C^{(S,Y)}$ with values in $\S^{d+1}_{+}$ (the set of nonnegative definite symmetric matrices), having $\fP$-q.s.\ continuous and nondecreasing paths, and which coincides $P$-a.s.\ with the usual  second characteristic $\br{(S,Y)^{c}}^{P}$ under each $P\in\fP$. Let $C^{S}$ be the $d\times d$ submatrix corresponding to $S$ and let $C^{SY}$ be the $d$-dimensional vector corresponding to the quadratic covariation of $S$ and $Y$. Setting $A_{t}:=\tr C^{S}_{t}$ to be the trace of $C^{S}$, we have $C^{S}\ll A$ $\fP$-q.s.\ and $C^{SY}\ll A$ $\fP$-q.s. Thus, we have $dC^{S} = c^{S}dA$ $\fP$-q.s.\ and $dC^{SY} = c^{SY}dA$ $\fP$-q.s.\ for the derivatives defined by
	\[
		c^{S}_{t}:= \tilde{c}^{S}_{t} \1_{\{\tilde{c}^{S}_{t}\in \S^{d}_{+}\}}, 
		\quad \tilde{c}^{S}_{t} := \limsup_{n\to\infty} \frac{C^{S}_{t} - C^{S}_{(t-1/n)\vee 0}} {A_{t} - A_{(t-1/n)\vee 0}}
  \]
  and
	\[
		c^{SY}_{t}:= \tilde{c}^{SY}_{t} \1_{\{\tilde{c}^{SY}_{t}\in \R^{d}\}}, 
		\quad \tilde{c}^{SY}_{t} := \limsup_{n\to\infty} \frac{C^{SY}_{t} - C^{SY}_{(t-1/n)\vee 0}} {A_{t} - A_{(t-1/n)\vee 0}},
  \]  
  where all operations are to be understood in a componentwise fashion (with $0/0:=0$, say). We observe that $c^{S}$ and $c^{SY}$ are $\F$-predictable.
  Let $(c^{S})^{\oplus}$ be the Moore--Penrose pseudoinverse of $c^{S}$. We define the $\F$-predictable process
	\[
	  H:= c^{SY} (c^{S})^{\oplus}
	\]	
	and show that it satisfies the requirements of the theorem.
	
	Fix $P\in\fP$. Recalling that $\fP$ contains all sigma martingale measures equivalent to $P$, we may apply the classical optional decomposition theorem in the form of \cite[Theorem~5.1]{DelbaenSchachermayer.99} for $P$ and the filtration $\F^{P}_{+}$ (which satisfies the usual assumptions) and we obtain that there exist an $\F^{P}_{+}$-predictable, $S$-integrable process $H^{P}$ and a nondecreasing process $K^{P}$ such that 
	\begin{equation}\label{eq:usualOptDecomp}
	  Y=Y_{0}+H^{P}\sint S - K^{P}\quad P\as
	\end{equation}
  Identifying the continuous local martingale parts with respect to $P$ (cf.\ \cite{JacodShiryaev.03}) on both sides of~\eqref{eq:usualOptDecomp} yields
	\begin{equation*}
	  Y^{c}=H^{P}\sint S^c\quad P\as
	\end{equation*}	
	and then taking quadratic covariation with $S^{c}$ leads to
	\begin{equation*}
	  d\br{S^{c},Y^{c}}=H^{P} d\br{S^c} \quad P\as
	\end{equation*}
	This, however, is equivalent to
	\begin{equation}\label{eq:equalityUnderA}
	  c^{SY}=H^{P} c^{S} \quad P\times dA\mbox{-a.e.}
	\end{equation}
	On the one hand, \eqref{eq:equalityUnderA} implies that $H$ is $S^{c}$-integrable under $P$ and
	\begin{equation}\label{eq:equalityScont}
	  H\sint S^{c} = H^{P}\sint S^{c}\quad P\as
	\end{equation}
	On the other hand, define the nondecreasing process $A^{*}$ by 
	\[
	  \frac{dA^{*}}{dA} = \min_{1\leq i\leq d} \frac{d(C^{S})^{ii}}{dA}\quad P\as
	\]
	As $A^{*}\ll A$, \eqref{eq:equalityUnderA} yields that
	\begin{equation*}
	  c^{SY}=H^{P} c^{S} \quad P\times dA^{*}\mbox{-a.e.}
	\end{equation*}
	It follows from the definition of $A^{*}$ that $c^{S}$ is invertible $P\times dA^{*}$-a.e., so we deduce that
	\begin{equation}\label{eq:optDecompHedgeIdentification}
    H=H^{P}\quad P\times dA^{*}\mbox{-a.e.} \quad \mbox{for all}\quad P\in\fP.
  \end{equation}
  By the dominating diffusion assumption, $A^{*}$ dominates the characteristics of $S-S^{c}$ under $P$, and so \eqref{eq:optDecompHedgeIdentification} implies that $H$ is $S-S^{c}$ integrable under~$P$ and $H\sint (S-S^{c})=H^{P}\sint (S-S^{c})$ $P$-a.s. In view of~\eqref{eq:equalityScont}, we conlcude that $H\sint S=H^{P}\sint S$ $P$-a.s.
  This holds for all $P\in\fP$, and now the theorem follows from~\eqref{eq:usualOptDecomp}.
\end{proof}

\begin{remark}\label{rk:sigmaMart}
  Theorem~\ref{th:optDecomp} remains true if we replace ``sigma martingale'' by ``local martingale'' throughout this section. The proof is the same; we merely need to replace the reference to \cite[Theorem~5.1]{DelbaenSchachermayer.99} by \cite[Theorem~1]{FollmerKabanov.98}.
\end{remark}


\section{Superreplication Duality}\label{se:duality}

In this section, we utilize Theorem~\ref{th:optDecomp} to provide a superhedging duality and an optimal superhedging strategy in a fairly general setting. Apart from the technical details (and of course the use of Theorem~\ref{th:optDecomp}), the line of argument is similar as in~\cite{NeufeldNutz.12}.

\subsection{Setting}

The optional decomposition theorem will be applied to a process $Y_{t}$ which is a version of the dynamic superhedging price. More precisely, this process will be constructed from the (conditional) sublinear expectation $\cE_{t}(f)$ of the claim $f$, and the $\fP$ supermartingale property of $Y$ will be deduced from the dynamic programming associated to $\cE_{t}(\cdot)$. In this section, we describe a setting where this can be implemented.

Let $\Omega=D_0(\R_+,\R^{d'})$ be the space of all c\`adl\`ag paths $\omega=(\omega_t)_{t\geq0}$ in $\R^{d'}$ with $\omega_0=0$, for some positive integer $d'$. We equip $\Omega$ with the Skorohod topology and the corresponding Borel $\sigma$-field $\cF$. Moreover, we denote by $\F=(\cF_t)_{t\geq0}$ the filtration generated by the canonical process $(t,\omega)\mapsto \omega_{t}$ and by
$\fP(\Omega)$ the space of probability measures on $(\Omega,\cF)$ equipped with the topology of weak convergence.

Given $t\geq0$, we introduce the following notation. The concatenation of two paths $\omega, \tilde{\omega}\in \Omega$ at $t$ is given by
\[
 (\omega\otimes_t \tilde{\omega})_u := \omega_u \1_{[0,t)}(u) + \big(\omega_{t} + \tilde{\omega}_{u-t}\big) \1_{[t, \infty)}(u),\quad u\geq 0.
\]
For any $P\in\fP(\Omega)$, there is a regular conditional
probability distribution $\{P^\omega_t\}_{\omega\in\Omega}$
given $\cF_t$ satisfying
\[
  P^\omega_t\big\{\omega'\in \Omega :\, \omega' = \omega \mbox{ on } [0,t]\big\} = 1\quad\mbox{for all}\quad\omega\in\Omega.
\]
We then define $P^{t,\omega}\in \fP(\Omega)$ by
\[
  P^{t,\omega}(F):=P^\omega_t(\omega\otimes_t F),\quad F\in \cF, \quad\mbox{where }\omega\otimes_t F:=\{\omega\otimes_t \tilde{\omega}:\, \tilde{\omega}\in F\}.
\]
Given a function $f$ on $\Omega$ and $\omega\in\Omega$,
we also define the function $f^{t,\omega}$  by
\[
  f^{t,\omega}(\tilde{\omega}) :=f(\omega\otimes_t \tilde{\omega}),\quad \tilde{\omega}\in\Omega.
\]
If $f$ is measurable, then $E^{P^{t,\omega}}[f^{t,\omega}]=E^P[f|\cF_t](\omega)$ for $P$-a.e.\ $\omega\in\Omega$. (The convention $\infty - \infty = -\infty$ is used; e.g., in defining the conditional expectation $E^P[f|\cF_t]:=E^P[f^+|\cF_t]-E^P[f^-|\cF_t]$.)

While we will eventually be concerned with a single set $\fP\subseteq \fP(\Omega)$ of measures, it is technically useful to induce $\fP$ by a family $\{\fP(s,\omega)\}_{(s,\omega)\in \R_+\times\Omega}$ of subsets of $\mathfrak{P}(\Omega)$ as follows. Let $\{\fP(s,\omega)\}$ be given and adapted in the sense that
\[
  \fP(s,\omega)=\fP(s,\omega')\quad\mbox{if}\quad \omega|_{[0,s]}=\omega'|_{[0,s]}.
\]
The set $\fP(0,\omega)$ is independent of $\omega$ as all paths start at the origin, and so we can define $\fP:=\fP(0,\omega)$. We assume throughout that $\fP\neq\emptyset$. In applications, $\fP$ will be the primary object and we specify a corresponding family $\{\fP(s,\omega)\}$ such that $\fP=\fP(0,\omega)$; see Section~\ref{se:Levy} for an example. Properties~(ii) and~(iii) of Condition~(A) below imply that the family $\{\fP(s,\omega)\}$ is essentially determined by the set $\fP$.

We recall that a subset of a Polish space is called analytic if it is the image of a Borel subset of another Polish space under a Borel-measurable mapping; in particular, any Borel set is analytic (see also \cite[Chapter~7]{BertsekasShreve.78} for background). We can now state the following condition on the measurability and stability under conditioning and pasting of $\{\fP(s,\omega)\}$; it will enable us to obtain the sublinear expectation $\cE_{t}(\cdot)$ and its dynamic programming.
 
\begin{conditionA}\label{condA}
  For all $0\leq s\leq t$, $\bomega\in\Omega$ and $P\in \fP(s,\bomega)$,
  \begin{enumerate}
    \item  $\{(P',\omega): \omega\in\Omega,\; P'\in \fP(t,\omega)\} \,\subseteq\, \mathfrak{P}(\Omega)\times\Omega$ is analytic;

    \item $P^{t-s,\omega} \in\fP(t,\bomega\otimes_s\omega)$ for $P$-a.e.\ $\omega\in\Omega$;

    \item if $\kappa: \Omega \to \mathfrak{P}(\Omega)$ is an $\cF_{t-s}$-measurable kernel and $\nu(\omega)\in \fP(t,\bomega\otimes_s\omega)$ for $P$-a.e.\ $\omega\in\Omega$,
    then the measure defined by
    \[
      \bar{P}(A)=\iint (\1_A)^{t-s,\omega}(\omega') \,\kappa(d\omega';\omega)\,P(d\omega),\quad A\in \cF
    \]
    is an element of $\fP(s,\bomega)$.
  \end{enumerate}
\end{conditionA}

Some more notation is needed to state the required result on $\cE_{t}(\cdot)$. Given a $\sigma$-field $\cG$, the universal completion of $\cG$ is the $\sigma$-field $\cG^*=\cap_P \cG^{(P)}$, where $P$ ranges over all probability measures on $\cG$ and $\cG^{(P)}$ is the completion of $\cG$ under $P$.  Moreover, a scalar function $f$ is called upper semianalytic if $\{f>a\}$ is analytic for all $a\in\R$. Any Borel-measurable function is upper semianalytic and any upper semianalytic function is universally measurable.

\begin{proposition}\label{pr:NvH}
   Let Condition~(A) hold, let $0\leq s\leq t$ and let $f:\Omega\to\overline{\R}$ be an upper semianalytic function. Then the function
   \[
     \cE_t(f)(\omega):=\sup_{P\in\fP(t,\omega)} E^P[f^{t,\omega}],\quad\omega\in\Omega
   \]
   is $\cF_t^*$-measurable and upper semianalytic. Moreover,
   \[
     \cE_s(f)(\omega) = \cE_s(\cE_t(f))(\omega)\quad\mbox{for all}\quad \omega\in\Omega.
   \]
   Furthermore, with $\fP(s;P)=\{P'\in \fP:\, P'=P \mbox{ on } \cF_s\}$, we have
   \begin{equation}\label{eq:esssupDPP}
     \cE_s(f) = \mathop{\esssup^P}_{P'\in \fP(s;P)} E^{P'}[\cE_t(f)|\cF_s]\quad P\as\quad\mbox{for all}\quad P\in\fP.
   \end{equation}
\end{proposition}

See \cite[Theorem~2.3]{NutzVanHandel.12} and the subsequent remark for this result. The theorem in \cite{NutzVanHandel.12} is stated for the space of continuous paths but carries over to Skorohod space without changes; only the Polish structure is important.

\subsection{Duality Result}

In what follows, we fix a set $\fP\subseteq \fP(\Omega)$ determined by a family $\{\fP(s,\omega)\}$ as above. We shall use the filtration $\G= (\cG_t)_{0\leq t \leq T}$, where
\begin{equation*}
  \cG_t:= \cF^{*}_{t}\vee \mathcal{N}^{\fP};
\end{equation*}
here $\cF^{*}_{t}$ is the universal completion of $\cF_t$ and $\mathcal{N}^{\fP}$ is the collection of sets which are $(\cF_T,P)$-null for all $P\in\fP$. Let $S$ be an $\R^{d}$-valued, $\G_{+}$-adapted process with c\`adl\`ag paths.
A $\G$-predictable process $H\in L(S,\fP)$ is called \emph{admissible} if
$H\sint S$ is a $P$-supermartingale for all $P\in\fP$, and we denote by $\cH$ the set of all such processes.

\begin{theorem}\label{th:duality}
  Suppose that $\{\fP(s,\omega)\}$ satisfies Condition~(A) and that $\fP$ is a nonempty, saturated set of sigma martingale measures for $S$ such that $S$ has dominating diffusion under all $P\in\fP$. Moreover, let
  $f:\Omega\to \overline{\R}$ be an upper semianalytic, $\cG_T$-measurable function such that  $\sup_{P\in\fP} E^P[|f|]<\infty$.
  Then
  \begin{align*}
    &\sup_{P\in\fP} E^P[f] \\
    &=\min\big\{x\in\R:\,\exists\, H\in \cH\mbox{ with } x+ H\sint S_{T} \geq f\;P\as\mbox{ for all }P\in\fP\big\}.
  \end{align*}
\end{theorem}

\begin{proof}
  As usual, one inequality in Theorem~\ref{th:duality} is immediate: if $x\in\R$ and there exists $H\in\cH$ such that $x+ H\sint S_{T}\geq f$, the supermartingale property of $H\sint S$ implies that $x\geq E^P[f]$ for all $P\in\fP$. Hence, we focus on showing that there exists $H\in\cH$ with
  \begin{equation*}
  \sup_{P'\in\fP} E^{P'}[f] + H\sint S_{T} \geq f \quad P \mbox{-a.s.} \quad \mbox{for all} \quad P \in \fP.
  \end{equation*}
  In view of Theorem~\ref{th:optDecomp}, we shall construct a c\`adl\`ag $\fP$ supermartingale $Y$ satisfying
  \begin{equation}\label{eq:aimY}
    Y_{0} \leq \sup_{P'\in\fP} E^{P'}[f] \quad \mbox{and} \quad Y_{T}=f \quad P\mbox{-a.s.} \quad \mbox{for all} \quad P \in \fP.
  \end{equation}
  
  Recall from Proposition~\ref{pr:NvH} that
  \[
    \cE_t(f)(\omega):=\sup_{P\in\fP(t,\omega)} E^P[f^{t,\omega}]
  \]
is $\cG_t$-measurable for all $t$. Our assumption that $\sup_{P\in\fP} E^P[|f|]<\infty$ implies that $\sup_{P\in\fP} E^P[|\cE_t(f)|]<\infty$; the argument is the same as in Step~1 of the proof of \cite[Theorem~2.3]{NeufeldNutz.12}.
Now~\eqref{eq:esssupDPP} implies that $\cE_t(f)$ is an $(\F^*,P)$-supermartingale for all $P\in\fP$. We can then define
\[
  Y'_{t}:= \limsup_{r\downarrow t,\, r\in \Q} \cE_{r}(f)\quad\mbox{for}\quad t<T \quad\mbox{and}\quad Y'_{T}:= \cE_{T}(f).
\]
The modification theorem for supermartingales \cite[Theorem~VI.2]{DellacherieMeyer.82} yields that outside a set $N\in\mathcal{N}^{\fP}$, the paths of $Y'$ are c\`adl\`ag and the limit superior is actually a limit, and moreover that $Y'$ is a $(\G_+,P)$-supermartingale for all $P\in\fP$. We define $Y:=Y'\1_{N^{c}}$; then all paths of $Y$ are c\`adl\`ag and $Y$ is still a $(\G_+,P)$-supermartingale for all $P\in\fP$ (recall that $\mathcal{N}^{\fP}\subseteq \cG_{0}$). In particular, $Y$ is a $\fP$ supermartingale relative to the filtration $\G_{+}$ in the terminology of Theorem~\ref{th:optDecomp}.

Fix $P \in \fP$; we check that $Y_{T}=f$ $P$-a.s. Indeed, we have $Y_{T}=Y'_{T}=\cE_{T}(f)$ $P$-a.s., and since $\cG_{T}=\cF_{T}$ $P$-a.s., \eqref{eq:esssupDPP} yields that 
$\cE_T(f)=f$ $P$-a.s. 

Next, we prove the first part of~\eqref{eq:aimY}; here the subtlety is that $Y_{0}$ need not be deterministic in general.  Fix again $P \in \fP$; we need to show that
\begin{equation}\label{eq:ineqTimeZero}
  Y_0 \leq \sup_{P'\in\fP} E^{P'}[f]\equiv \cE_0(f) \quad P \mbox{-a.s.}
\end{equation}
Let $y^{P}_0\in\R$ denote the least constant dominating $Y_{0}$ $P$-a.s.; then the above is of course equivalent to
\[
  y^{P}_0 \leq \cE_0(f).
\]
Let $P'\in\fP$. From the definition of $Y$ and \cite[Theorem~VI.2]{DellacherieMeyer.82}, we have that 
\[
 E^{P'}[Y_0|\cF_0]\leq \cE_0(f) \quad P'\as,
\]
but as $\cF_0=\{\emptyset,\Omega\}$, both sides are identically equal to real numbers and thus
$
 \sup_{P'\in\fP}E^{P'}[Y_0]\leq \cE_0(f)
$
as $P'\in\fP$ was arbitrary. 
As a result, it suffices to show that 
\[
  y^{P}_0\leq \sup_{P'\in\fP}E^{P'}[Y_0].
\]
To this end, note that $y^{P}_0 =\sup_{Q}E^{Q}[Y_0]$, where the supremum is taken over all probability measures $Q$ on $\cF_{0+}$ which are equivalent to $P$. Hence, it suffices to establish that each such $Q$ is the restriction to $\cF_{0+}$ of some element $P'$ of $\fP$. Indeed, fix $Q$ and let $Z=dQ/dP$ be a Radon--Nikodym derivative with respect to $\cF_{0+}$. We define the measure $P'$ on $\cG_{T}$ by
$dP'=Z dP$; then $Q=P'|_{\cF_{0+}}$. Moreover, using that $Z$ is $\cF_{0+}$-measurable and that $S$ is right-continuous, we see that $P'$ is again a sigma martingale measure for $S$ and thus $P'\in\fP$ by the assumed saturation of $\fP$. We have shown the inequality~\eqref{eq:ineqTimeZero}, and this completes the construction of the $\fP$ supermartingale $Y$ satisfying~\eqref{eq:aimY}.

Theorem~\ref{th:optDecomp}, applied with the $\sigma$-field $\cG_{T}$ and the filtration $\G_{+}$, yields a $\G_{+}$-predictable process $H\in L(S,\fP)$ such that 
\[
  \sup_{P'\in\fP} E^{P'}[f] + H\sint S_{T} \geq Y_{0} + H\sint S_{T} \geq Y_{T} = f \quad P\mbox{-a.s.} \quad \mbox{for all} \quad P \in \fP.
\]
We recall that $\G_{+}$-predictability is the same as $\G$-predictability because any $\G_{+}$-adapted left-continuous process is also $\G$-adapted. To see that $H$ is admissible, we note that for every $P\in\fP$, the sigma martingale $H\sint S$ is $P$-a.s.\ bounded from below by the martingale $E^P[f|\G]$; this implies the supermartingale property by Fatou's lemma.
\end{proof}


\section{Application to Nonlinear L\'evy Processes}\label{se:Levy}

In this section, we give a natural example for the above theory in the context of nonlinear L\'evy processes. This model was first introduced by \cite{HuPeng.09levy} with a construction based on partial integro-differential equations (PIDE), thus extending the notions of \cite{Peng.07, Peng.08} to jump processes. A more general construction of nonlinear L\'evy processes was provided in \cite{NeufeldNutz.13b}.

We continue to use the canonical setup $\Omega=D_0(\R_+,\R^{d'})$ of the preceding section, and we choose $S$ to be the canonical process $S_{t}(\omega)=\omega_{t}$ (i.e., $d=d'$). Let $\fP_{sem}$ be the set of all probabilities $P$ under which $S$ is a seminartingale (again, we need not be specific about the filtration; cf.\ \cite[Proposition~2.2]{NeufeldNutz.13a}). We shall focus on the subset
\begin{equation*}
  \fP^{ac}_{sem}=\big\{P\in \fP_{sem}:\, (B^P,C^P,\nu^P)\ll dt, \; P\as\big\}
\end{equation*}
of semimartingales having absolutely continuous characteristics with respect to the Lebesgue measure $dt$.
Given $P\in\fP^{ac}_{sem}$, we can consider the associated differential characteristics $(b^P,c^P, F^P)$ defined via $(dB^P,dC^P,d\nu^P)=(b^P dt, c^Pdt, F^Pdt)$. The differential characteristics take values in $\R^d \times \S^d_+\times\cL$, where $\S^d_+$ is the set of symmetric nonnegative definite $d\times d$-matrices and
\begin{equation*}
  \cL= \bigg\{ F \mbox{ measure on } \R^d: \, \int_{\R^d} |x|^2 \wedge 1 \, F(dx) <\infty \ \mbox{and} \ F(\{0\})=0  \bigg\}
\end{equation*}
is the set of all L\'evy measures (a separable metric space under a suitable weak convergence topology; cf.\ \cite[Section~2]{NeufeldNutz.13a}). An element $(b,c,F)\in \R^d \times \S^d_+\times\cL$ is called a L\'evy triplet, and we recall that for every such triplet there exists a L\'evy process having $(b,c,F)$ as its differential characteristics.

Let $\emptyset\neq\Theta\subseteq \R^d \times \S^d_+\times\cL$ be a collection of L\'evy triplets, then we can consider the set of all semimartingale laws whose differential characteristics evolve in~$\Theta$,
\[
  \fP_\Theta:=\big\{P\in \fP^{ac}_{sem}:\, (b^P,c^P,F^P)\in \Theta, \,P\otimes dt\mbox{-a.e.}\big\}.
\]
This will be our basic set $\fP$ in the sequel; indeed, by \cite[Theorem~2.1]{NeufeldNutz.13b}, we have the following fact. 

\begin{lemma}\label{le:levyCondA}
  Let $\Theta\subseteq \R^d \times \S^d_+\times\cL$ be Borel-measurable and $\fP(s,\omega):=\fP_\Theta$ for all $(s,\omega)\in \R_+\times\Omega$. Then the collection $\{\fP(s,\omega)\}$ satisfies Condition~(A).
\end{lemma}

To wit, in the present example, $\fP(s,\omega)$ does not depend on $(s,\omega)$ at all; this reflects the fact that the L\'evy process is homogeneous in time and space. The dependence would be nontrivial, e.g., in the context of a controlled stochastic differential equation as in~\cite{Nutz.11} or for the random $G$-expectation as in~\cite{Nutz.10Gexp, NutzVanHandel.12}.

Next, we determine when $\fP_\Theta$ satisfies the conditions of our main results. To this end, we introduce the set of L\'evy measures with integrable jumps,
\[
  \cL^{*} = \left\{ F\in \cL:\, \int (|x|^{2}\wedge |x|)\, F(dx)<\infty\right\}
\]
and denote by $\S^{d}_{++}\subseteq \S^{d}_{+}$ the set of strictly positive definite matrices.

\begin{lemma}\label{le:recipe}
	Let $\Theta\subseteq \R^d \times \S^d_+\times\cL$. The set $\fP_\Theta$ consists of sigma martingale measures for $S$ if and only if $\Theta$ is of the form 
  \begin{equation}\label{eq:formTheta}
    \Theta = \left\{(b,c,F)\in \R^d \times \S^d_+\times\cL:\, (c,F)\in \Theta',\, b = \int (x-h(x))\, F(dx) \right\}
  \end{equation}
  for some set $\Theta'\subseteq \S^{d}_{+}\times \cL^{*}$. 
  In this case, $S$ has dominating diffusion under all $P\in\fP_\Theta$ if and only if
  \begin{equation}\label{eq:diffusionTheta}
    \Theta' \; \subseteq \; (\S^d_+\times \{0\}) \cup  (\S^d_{++}\times\cL^{*}).
  \end{equation}
  Moreover, the following condition is sufficient for $\fP_{\Theta}$ to be saturated:
  \begin{equation}\label{eq:saturationCond}
    \mbox{$(c,F)\in\Theta'$ implies $(c,\psi F)\in\Theta'$ for all $\psi>0$ such that $\psi F\in\cL^{*}$},
  \end{equation}
  where $\psi: \R^{d}\to (0,\infty)$ is Borel-measurable.
\end{lemma}

We remark that $\Theta$ is nonempty and Borel-measurable if and only if $\Theta'$ is. The condition~\eqref{eq:saturationCond} is not quite necessary for saturation; for instance, with $\Theta=\{(0,0,\delta_{1})\}$ we obtain the set corresponding to the compensated Poisson process, which is saturated but violates~\eqref{eq:saturationCond}. However, an arbitrary change of intensity is clearly possible as soon there is a nonvanishing Brownian component, so that~\eqref{eq:saturationCond} is not that far from being necessary.

\begin{proof}
  By \cite[Proposition~III.6.35, p.\,215]{JacodShiryaev.03}, $S$ is a sigma martingale under $P\in\fP^{ac}_{sem}$ if and only if its differential characteristics $(b^{P}_{t},c^{P}_{t},F^{P}_{t})$ satisfy
  \begin{equation}\label{eq:sigmaMartCharact}
    F^{P}_t\in\cL^{*} \quad \mbox{and}\quad  b^{P}_{t} + \int (x-h(x))\, F^{P}_{t}(dx)=0, \quad P\times dt\mbox{-a.e.}
  \end{equation}
  The ``if'' statement in the first claim follows immediately, whereas for the reverse assertion we use the existence of L\'evy processes with arbitrary triplets.
 
  It is straightforward to see that~\eqref{eq:diffusionTheta} implies dominating diffusion. Conversely, if the latter holds, then for any L\'evy law $P\in\fP_{\Theta}$ we must have either $F^{P}=0$ (and then $c^{P}$ can be arbitrary), or $(|x|^{2}\wedge 1) \ast \nu^{P}_{t} =  (\int (|x|^{2}\wedge 1)\,dF^{P}) t$ is strictly increasing and then $c^{P}$ must be positive.

  Finally, we show that~\eqref{eq:saturationCond} is sufficient for saturation. Indeed, let $P\in\fP_{\Theta}$  and  $P'\sim P$. Then the general Girsanov theorem \cite[Proposition~III.3.24, p.\,172]{JacodShiryaev.03} shows that $P'\in\fP^{ac}_{sem}$ and the corresponding characteristics satisfy $c^{P'}=c^{P}$ and $F^{P'}=\psi F^{P}$ for some positive predictable function $\psi=\psi(\omega,t,x)$. If $P'$ is a sigma martingale measure for $S$, then $\psi F^{P}\in \cL^{*}$ $P\times dt$-a.e.\ by~\eqref{eq:sigmaMartCharact} and now~\eqref{eq:saturationCond} yields that $P'\in\fP_{\Theta}$.  
\end{proof}

As a consequence of the two preceding lemmas, we have the following result.

\begin{proposition}\label{pr:LevyWorks}
  Let $\emptyset\neq \Theta\subseteq \R^d \times \S^d_+\times\cL$ be Borel-measurable, of the form~\eqref{eq:formTheta} and satisfy~\eqref{eq:diffusionTheta} and~\eqref{eq:saturationCond}. Then Theorem~\ref{th:duality} applies to $\fP=\fP_{\Theta}$.
\end{proposition}

\begin{remark}\label{rk:sigmaMartUseful}
  (i) The present section illustrates that it is convenient to work with sigma martingales rather than local martingales. Indeed, even though a L\'evy sigma martingale (in the classical sense) is automatically a local (even true) martingale, it may easily happen that $\fP_{\Theta}$ contains sigma martingale measures which are not local martingale measures. To see this, recall that local martingales are described by satisfying, in addition to~\eqref{eq:sigmaMartCharact}, the condition that $\int_{0}^{t}\int (|x|^{2}\wedge |x|)\, F_{t}(dx)dt<\infty$ $P$-a.s.\ for all $t$. With a view towards Condition~(A), this property would be inconvenient to deal with.

  (ii) In our context of superhedging, it is crucial to use the general construction of nonlinear L\'evy processes as in \cite{NeufeldNutz.13b} rather than the PIDE-based construction of \cite{HuPeng.09levy}, even in the realm of finite variation jumps and continuous claims. Indeed, the saturation condition (which is vital) typically leads to a violation of the compactness conditions which ensure a well-defined PIDE in \cite{HuPeng.09levy}. For instance, if we want to consider for $S$ the sum of a Wiener process and a Poisson process, saturation imposes that we must allow for unbounded intensity changes in the jump part.
\end{remark}

\subsection{Exponential L\'evy Processes}

In the context of stock price modeling in finance, the process $S$ is often of exponential form; see, e.g., \cite{NeufeldNutz.15}. We continue with the same setting as above, except that we write~$X$ for the canonical process $X_{t}(\omega)=\omega_{t}$ and aim to define $S$ as the stochastic exponential of $X$.

More precisely, let $\fP_\Theta$ be a set of sigma martingale measures for $X$. The exponential stochastic differential equation
\begin{equation}\label{eq:exponential}
  dS=\diag (S_{-})\,dX,\quad S_{0}=1
\end{equation}
can be solved ``pathwise'' as in \cite{Karandikar.95}, or alternatively via the Dol\'eans--Dade formula \cite[I.4.64, p.\,59]{JacodShiryaev.03} and the arguments in the proof of Theorem~\ref{th:optDecomp}: recalling that $\mathcal{N}^{\fP}\subseteq \cG_{0}$ and using the fact that $X$ is c\`adl\`ag, we can construct a c\`adl\`ag, $\G_{+}$-adapted (even $\G$-adapted) process $S$ which solves~\eqref{eq:exponential} under all $P\in\fP_\Theta$. In order to have positive prices, we choose $\Theta$ such that the contained L\'evy measures are concentrated on $(-1,\infty]^{d}$, which is equivalent to having $\Delta X^{i} > -1$ $\fP_{\Theta}$-q.s.\ and also to $S^{i} > 0$ and $S^{i}_{-} > 0$ $\fP_{\Theta}$-q.s.
 Then, $P$ is a sigma martingale measure for $X$ if and only if it is a sigma martingale measure for $S$, and the characterization~\eqref{eq:diffusionTheta} for the dominating diffusion property remains valid. As a result, Lemma~\ref{le:recipe} carries over to the exponential case without changes, and so does Proposition~\ref{pr:LevyWorks}. We remark that in this situation, there is no difference between sigma martingale measures and local martingale measures, as any positive sigma martingale is a local martingale.


\newcommand{\dummy}[1]{}

\end{document}